\documentclass[11pt, letter]{article}
\usepackage[T1]{fontenc}
\usepackage[latin9]{inputenc}
\usepackage{geometry}
\usepackage[active]{srcltx}

\makeatletter
\pdfoutput=1

\usepackage{amsmath, amsthm, amssymb, amsfonts}
\usepackage[basic]{complexity}
\usepackage[pdftex,bookmarks,colorlinks]{hyperref}
\usepackage{enumitem}
\usepackage[usenames,dvipsnames]{color}
\usepackage{tikz}
\usepackage{todonotes}

\hypersetup{colorlinks=true,citecolor=red}

\usepackage{fullpage,subfigure}
\usepackage{url}
\usepackage[linesnumbered,lined,ruled,noend]{algorithm2e}

\newcommand{\be}{\begin{equation}}
\newcommand{\ee}{\end{equation}}
\newcommand{\coreset}{\textsc{Coreset}}

\newcommand{\eps}{\ensuremath{\varepsilon}}                       
\renewcommand{\epsilon}{\varepsilon}
\newcommand{\REAL}{\ensuremath{\mathbb{R}}}                       

\newcommand{\SamplesNeeded}{\BigO_{\epsilon,\delta}\left(\log n (d\log \log n + \log^2 n) \right)}

\newcommand\tab[1][0.25cm]{\hspace*{#1}}
\newcommand\Reals{\mathbb{R}}

\renewcommand\PP{\mathcal{P}}

\renewcommand\SS{\mathcal{S}}
\newcommand\BigO{\mathcal{O}}

\newtheorem{theorem}{Theorem}
 \newtheorem{corollary}[theorem]{Corollary}
 \newtheorem{lemma}[theorem]{Lemma}
 
\newtheorem{definition}[theorem]{Definition}
\newtheorem{assumption}[theorem]{Assumption}

\title{Training Support Vector Machines using Coresets}

\author{Cenk Baykal\\
MIT\\
baykal@mit.edu\and Lucas Liebenwein\\
MIT\\
lucasl@mit.edu\and  Wilko Schwarting\\
MIT\\
wilkos@mit.edu}

\begin{document}
\date{}
\maketitle

\begin{abstract}
\textbf{Note: This work was done as a course project as part of an ongoing research effort that was recently submitted \cite{anonymous2018small}. The submission, done in collaboration with Murad Tukan, Dan Feldman, and Daniela Rus \cite{anonymous2018small}, supersedes the work in this manuscript.}





We present a novel coreset construction algorithm for solving classification tasks using Support Vector Machines (SVMs) in a computationally efficient manner. A coreset is a weighted subset of the original data points that provably approximates the original set. We show that coresets of size $\SamplesNeeded$, i.e., polylogarithmic in $n$ and polynomial in $d$, exist for a set of $n$ input points with $d$ features and present an $(\epsilon,\delta)$-FPRAS for constructing coresets for scalable SVM training. Our method leverages the insight that data points are often redundant and uses an importance sampling scheme based on the sensitivity of each data point to construct coresets efficiently. We evaluate the performance of our algorithm in accelerating SVM training against real-world data sets and compare our algorithm to state-of-the-art coreset approaches. Our empirical results show that our approach outperforms a state-of-the-art coreset approach and uniform sampling in enabling computational speedups while achieving low approximation error.
\end{abstract}

\section{Introduction}
\label{sec:Introduction}
Popular machine learning algorithms are computationally expensive, or worse yet, intractable to train on Big Data. Recently, the notion of using \textit{coresets} \cite{agarwal2005geometric,feldman2011unified,braverman2016new}, small weighted subsets of the input points that approximately represent the original data set, has shown promise in accelerating machine learning algorithms, such as $k$-means clustering \cite{feldman2011unified}, training mixture models \cite{feldman2011scalable}, and logistic regression \cite{huggins2016coresets}. Support Vector Machines (SVMs) are one of the most popular algorithms for classification and regression analysis. However, with the rising availability of Big Data, training SVMs on massive data sets has shown to be computationally expensive. In this paper, we present a coreset construction algorithm for efficiently training Support Vector Machines.

Our approach entails a randomized coreset construction that is based on the insight that data is often redundant and that some input points are more important than others for large-margin classification. Using importance sampling, our algorithm can be considered an $(\epsilon,\delta)$-FPRAS which generates a coreset that could be used for training instead of the original (massive) set of input points, but yet still provide an $\epsilon$-approximation to the ground-truth classifier if all the points were used instead, with probability at least $1 - \delta$. In this paper, we prove that such coresets of size $\BigO_{\epsilon,\delta}\left(\log n (d\log \log n + \log^2 n)) \right)$ can be efficiently constructed for a set of $n$ points with $d$ features and present an intuitive, importance sampling-based approach for constructing them. 

\section{Related Work}
\label{sec:RelatedWork}
Training a canonical Support Vector Machine (SVM) requires $\BigO(n^3)$ time and $\BigO(n^2)$ space \cite{tsang2005core} where $n$ is the number of training points, which may be impractical for certain applications. Work by Tsang et al.~\cite{tsang2005core} investigated computationally-efficient approximations in terms of coresets to the SVM problem, termed Core Vector Machines (CVMs), and leveraged existing coreset methods for the Minimum Enclosing Ball (MEB) \cite{agarwal2005geometric,badoiu2003smaller}. The authors propose a method that reduces the training time required for the two-class L2-SVM to $\mathcal{O}(n)$ and the space to an expression that is (surprisingly) independent of $n$. Similar geometric approaches based on convex hulls and extreme points were investigated by \cite{nandan2014fast}.


Since the SVM problem is inherently a quadratic optimization problem, prior work has investigated approximations to the quadratic programming problem using the Frank-Wolfe algorithm or Gilbert's algorithm \cite{clarkson2010coresets}. Another line of research has been in reducing the problem of polytope distance to solve the SVM problem \cite{gartner2009coresets}. The authors establish lower and upper bounds for the polytope distance problem and use Gilbert's algorithm to train an SVM in linear time. 

Har-Peled et al.\ constructed coresets to approximate the maximum margin separation, i.e., a hyperplane that separates all of the input data with margin larger than $(1-\epsilon)\rho^*$, where $\rho^*$ is the best achievable margin \cite{har2007maximum}. They study the running time of a simple coreset algorithm for binary and ``structured-output'' classification and the use of coresets for active learning and noise-tolerant learning in the agnostic setting.

There have been probabilistic approaches to the SVM problem. Most notable are the works of Clarkson et al.\ \cite{clarkson2012sublinear} and Hazan et al \cite{hazan2011beating}. Hazan et al. used a primal-dual approach combined with the Stochastic Gradient Descent (SGD) approach in order to learn linear SVMs in sublinear time. They propose the SVM-SIMBA approach which returns an $\epsilon$-approximate solution with probability at least $1/2$ to the SVM problem that uses Hinge loss as the objective function. The key idea in their method is to access single features of the training vectors rather than the entire vectors themselves. However, their method is nondeterministic and returns the correct $\epsilon$-approximation only with some probability greater than a constant ($1/2$).
Clarkson et al.~\cite{clarkson2012sublinear} present sublinear-time (in the size of the input) approximation algorithms for some optimization problems such as training linear classifiers (e.g., perceptron) and finding MEB. They introduce a technique that is  originally applied to the perceptron algorithm, but extend it to the related problems of MEB and SVM in the hard margin or $L2$-SVM formulations. The drawback of their method is that the approximation can be successfully computed only with high probability. Pegasos~\cite{shalev2007pegasos} employs primal estimated Subgradient Descent independent of the input size.

Joachims presents an alternative approach to training SVMs in linear time based on the cutting plane method that hinges on an alternative formulation of the SVM optimization problem \cite{joachims2006training}. He shows that the Cutting-Plane algorithm can be leveraged to train SVMs in $\BigO(sn)$ time for classification and $\BigO(sn \log n)$ time for ordinal regression where $s$ is the average number of non-zero features. However, this approach does not trivially extend to SVMs with kernels and is not sublinear with respect to the number of points $n$.

\section{Problem Definition}
\label{sec:problem-definition}
Given a set of training points $\PP = \{(x_1, y_1) \ldots, (x_n, y_n) \}$ where $x_i \in \PP \subseteq \Reals^d$ and $y_i \in \{-1,1\}$ for all $i \in [n]$, the soft-margin two-class L2-SVM problem is the following quadratic program:
\begin{align}
\min_{w \in Q(\PP)} &\tab f(\PP, w),
\end{align}
where 
\begin{equation}
\label{eqn:svm}
f(\PP, w) =  ||w||_2^2 + C \sum_{i \in [n]} \max\{0, 1 - y_iw^Tx_i\},
\end{equation}
for some regularization parameter $C \in \REAL_+$ and query space $Q(\PP)$, which is defined to be the set of all candidate margins. We note that we ignore the bias term $b$ in the problem formulation for simplicity, since this term can always be embedded into the feature space by expanding the dimension to $d+1$.

Instead of establishing new algorithms for the SVM problem, we focus on reducing the size of the training data by sampling \emph{informative points}, while providing each point a proper weight. 

\subsection{Soft-margin SVM}
\begin{figure}[h!]
 \centering
  \includegraphics[width=0.5\textwidth]{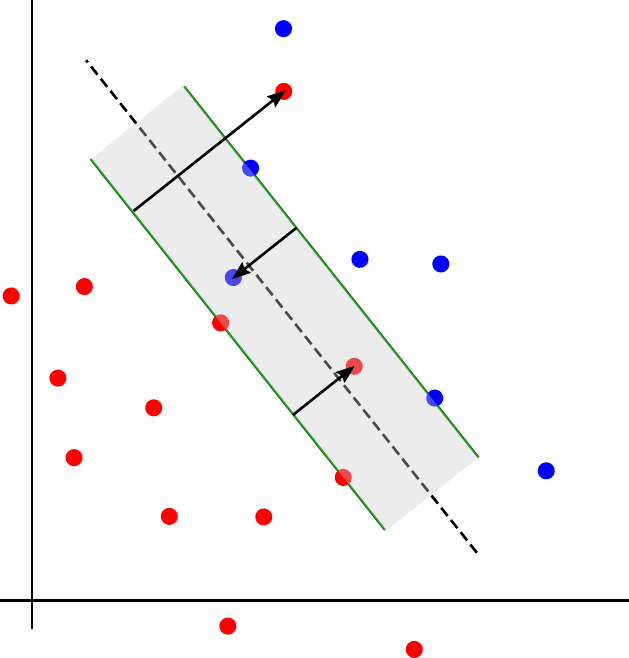}
 \caption{Soft-margin SVM}
 \label{fig:softmargin}
 \end{figure}
If the training data are linearly separable, we can select two parallel hyperplanes that separate the two classes of data, so that the distance between them $\rho = 2/||w||$ is as large as possible. The region bounded by these two hyperplanes is called the margin.
To extend to cases in which the data are not linearly separable a hinge loss function is introduced that penalizes the violation of the margin constraints. This approach is known as soft-margin SVM. Each given data point $x_i$ falls in one of three categories. It either lies beyond the margin $y_i w^Tx_i>1$, in which case it does not contribute to the SVM loss \eqref{eqn:svm}. The point could lie directly on the margin $y_i w^Tx_i=1$, where the point is a support vector and directly affects the cost function but does not directly add to it. Note that the distance between the separating hyperplane and the points on the margin is exactly $1/||w||$. This distance is commonly referred to as \emph{margin} as well. If the point lies within the margin $y_i w^Tx_i<1$ it adds a cost to \eqref{eqn:svm} proportional to the amount of constraint violation.

The regularization parameter $C$ weights the relative importance of maximizing the margin and margin constraint satisfaction for each data point $x_i$ and is used to increase robustness against outliers. Accordingly, if C is very small, margin constraint violation is only penalized weakly, $||w||$ is small and the safety margin around the decision boundary will be large. Contrary, if C is very large, violation of the margin constraint is penalized heavily and the formulation approaches the hard-margin SVM case which is sensitive to outliers in the training set.

\subsection{Coresets}
Instead of introducing an entirely new algorithm for solving the SVM problem itself, we use (smally) subsets of the input data instead for training, i.e., coresets, which 
Using coresets, the main benefit is that we can reduce the runtime through reduction of the number of training data points, while maintaining a close approximation to the optimal solution of the problem. The following definitions for coresets are used throughout the paper: 

\begin{definition}[Query Space]
\label{querySpace}
Let $\PP$ be the set of input points, let $Q(\PP)$ denote the set of possible margins over which the SVM optimization is performed over, and let $f$ be the SVM objective function given by \eqref{eqn:svm}. Then, the tuple $\mathcal{F} = (\PP, Q(\PP), f)$ is called the \emph{query space}.
\end{definition}

\begin{definition}[$\eps$-coreset] 
\label{coreset}
Let $(\SS,u)$ be a weighted subset of the input points $\PP$ such that the function $u: \PP \to \mathbb{R}_{\ge 0}$ maps each point to its corresponding weight. The pair $(\SS,u)$ is called a \emph{coreset} with respect to the input points $\PP$.
Let $Q(\PP)$ be the \emph{query set}, i.e., the set of candidate margins, and let $f$ be the SVM objective function given by \eqref{eqn:svm}. Then $(\SS,u)$ is an $\eps$-coreset if for every margin $w \in Q(\PP)$ we have
$$
|f(\PP,w)-f((\SS,u), w)|\leq \eps f(\PP,w).
$$
\end{definition}
In other words, $(\SS,u)$ is a $\eps$-coreset if $f((\SS,u), w)$ is an $1 \pm \epsilon$ approximation to the objective function value with all of the training points used, $f(\PP, w)$. Our overarching goal is to construct an $\eps$-coreset $\SS \subset \PP$ such that the cardinality of $\SS$ is sublinear in $n$, the number of data points. In our analysis, we will also rely on the concept of sensitivity $s(p_i)$ of a point $p_i$, see definition below, which  has been previously introduced in~\cite{braverman2016new}:

\begin{definition}[Sensitivity]
\label{def:sensitivity}
The sensitivity of an arbitrary point $p_i = (x_i,y_i)$ is defined as
\begin{equation}
\forall{i \in [n]} \tab s(p_i) = \sup_{w \in Q(\PP)} \,\, \frac{\tilde{f}(p_i, w)}{\sum_{j \in [n]} \tilde{f}(p_j, w)},
\end{equation}
where 
$$
\tilde{f}(p_i, w) = \frac{||w||_2^2}{n} + C \max\{0, 1-y_i w^T x_i\}.
$$
\end{definition}

Note that $\tilde{f}(p_i, w)$ represents the \textit{contribution} of point $p_i$ to the objective function of the SVM and that
\begin{align*}
\sum_{j \in [n]} \tilde{f}(p_j, w) = ||w||_2^2 + C\sum_{j \in [n]} \max\{0, 1 - w^T x_j y_j\}  = f(\PP, w)
\end{align*}
where $f(\PP, w)$ is the objective function of the two class L2-norm SVM as in \eqref{eqn:svm}.

\section{Analysis}
\label{sec:Analysis}

In this section, we prove under mild assumptions that an $\epsilon$-coreset of size $\BigO_{\epsilon,\delta}\left(\log n (d\log \log n + \log^2 n) \right)$ can be constructed with probability at least $1-\delta$, in $\BigO(nd)$ time. At a high level, our algorithm first efficiently approximates the importance of each point $p_i \in \PP$, which we refer to as the point's \emph{sensitivity} $s(p_i)$. The number of sample points required for an $(\epsilon, \delta)$-approximation is then computed as a function of the points' sensitivities using an analogue of the Estimator Theorem covered in class (and by Motwani et al. \cite{motwani2010randomized}), i.e., Theorem~\ref{thm:sampling} by Braverman et al. \cite{braverman2016new}. 


The outline of our proof is as follows. We begin by enumerating the preliminary material as well as the assumptions we impose on the problem. We then bound the sensitivity of each point by computing a tight upper bound that can be efficiently computed for all points. We then sum over all the upper bounds for the sensitivities of the points and show this sum is logarithmic in the number of points $n$. We then invoke Theorem~\ref{thm:sampling} with the computed sum of sensitivities and a straightforward application of the theorem's expression for the number of points required yields the existence of an $\epsilon$-coreset polylogarithmic in the number of points $n$ and polynomial in the number of features $d$. Combining these procedures, we finally arrive at the $(\epsilon, \delta)$-FPRAS, shown as Alg.~\ref{algorithm}.



\subsection{Preliminaries}
\label{sec:Preliminaries}
In the following, we state some assumptions and results upon which, we base our analysis.
\begin{assumption}[Normalized Input]
\label{asm:normalizedInput}
The training data is normalized such that for any $p_i = (x_i, y_i)$, $i \in [n]$ we have $||x_i||_2 \leq 1$.
\end{assumption}
\begin{assumption}[Scaled Input]
\label{asm:centeredInput}
The training data is centered around its mean $\mu = \frac{1}{n} \sum_{i=1}^n x_i$, such that $\mu = 0$.
\end{assumption}
Assumptions~\ref{asm:normalizedInput} and ~\ref{asm:centeredInput} are very commonly fulfilled in practical settings since both normalization and mean-centering of the input points are desirable before the training procedure for more robust results.
\begin{assumption}[Bounded Query Space]
\label{asm:boundedQuerySpace}
Let $\mathcal{W} = \{w \in \mathbb{R}^d: ||w||_2 \leq \log n\}$. The query set $Q(\PP)$ is then defined to be the set
\begin{equation}
Q(\PP) = \big\{w \in \mathcal{W} : \sum_{i \in [n]} \max \{0, 1 - w^T x_i y_i\} \geq n/\log n\big\}.
\end{equation}
\end{assumption}
In other words, we consider the set of candidate margins that do not entirely separate the labeled data, as is usually the case in target coresets applications that involve an extremely large number of data points\footnote{In future work, we intend to relax this assumption in our sensitivity analysis by leveraging a probabilistic argument in conjunction with the fact that data points are centered.}.

Assumption~\ref{asm:boundedQuerySpace} of having a bounded query space is justified by the fact that the points lie within a unit ball and therefore margins in accordant scale are reasonable. Moreover, we note that in many coreset applications a bounded space is a necessary condition for having coresets of sublinear size as shown for the case of Logistic Regression \cite{huggins2016coresets}.

Our analysis further relies on Theorem~\ref{thm:sampling} given by Braverman et al.~\cite{braverman2016new}, which is stated below. This result states that for any given overapproximation $\gamma(p_i)$ of the sensitivity $s(p_i)$, a coreset of sufficiently large size, where the size depends on the tightness of the overapproximation, gives an $\eps$-coreset with probability at least $1-\delta$. This theorem, together with our subsequent analysis, will allow us to establish the aforementioned $(\eps,\delta)$-FPRAS for computing the margin of a SVM classifier.

\begin{theorem}[Braverman et al. \cite{braverman2016new}]
\label{thm:sampling}
Let $\gamma: P \to \Reals_+$ be a function such that
$$
\forall{i \in [n]} \tab \gamma(p_i) \ge s(p_i),
$$
and let $t = \sum_{i \in [n]} \gamma(p_i)$. Further, let $\mathcal{F} = \left(\PP, Q(\PP), f \right)$ denote the query space and let $\text{dim}(\mathcal{F})$ be the corresponding VC dimension \cite{vapnik1998statistical}. Then, for all $\epsilon \in (0, 1)$, there exists some sufficiently large constant $c \ge 1$ such that for a random sample $\SS \subset \PP$ of size 
\begin{equation}
\label{eqn:samplesNeeded}
|\SS| \ge \frac{ct}{\epsilon^2}\left(dim(\mathcal{F})\log(t) + t\log(\frac{1}{\delta})\right)
\end{equation}
we have that $p = q$ with probability $s(p)/t$ for every $p \in \PP$ and $q \in S$. Let 
\begin{equation}
u(p) = \frac{K_p t (p)}{s(p) |\SS|}
\end{equation} be the weight for every $p \in \SS$, where $K_p$ is the number of times point $p$ is sampled. Then, with probability at least $1 - \delta$, $(S,u)$ is an $\epsilon$-coreset for $\PP$.  
\end{theorem}

\subsection{Sensitivity Upper Bound}
\label{sec:sensitivityUpperBounds}
To be able to derive an $(\eps,\delta)$-FPRAS according to Theorem~\ref{thm:sampling}, we need to be able to efficiently and tightly upperbound the sensitivity $s(p_i)$ of each point. In particular, we start out from the sensitivity description in its most basic form and use multiple insights from SVM. 

\begin{lemma}[Sensitivity Bounds]
\label{lem:sensitivityBounds}
The sensitivity of any arbitrary point $p_i \in \PP$ is bounded above by
\begin{equation}
s(p_i) \leq \gamma(p_i) = \frac{1}{n} + \frac{\log n  + ||x_i||_2 \log^2 n}{n}.
\end{equation}
\end{lemma}
\begin{proof}
Consider the sensitivity of a particular point $p_i \in \PP$:
\begin{align}
s(p_i) &= \sup_{w \in Q(\PP)} \,\, \frac{\tilde{f}(p_i, w)}{\sum_{j \in [n]} \tilde{f}(p_j, w)} \\
&\leq\sup_{w \in Q(\PP)} \,\, \frac{1}{n} + \frac{\max \{0, 1 - w^T x_i y_i\}}{ \sum_{j \in [n]} \max \{0, 1 - w^T x_j y_j\}} && \left(\frac{a + b}{c + d} \leq \frac{a}{c} + \frac{b}{d}, \forall{a,b,c,d \in \mathbb{R}_+}\right) \\
&\leq\sup_{w \in Q(\PP)} \,\, \frac{1}{n} + \frac{1 + ||w||_2 ||x_i||_2}{ \sum_{j \in [n]} \max \{0, 1 - w^T x_j y_j\}} && \text{(Cauchy-Schwarz Inequality)} \\
&\leq \frac{1}{n} + \frac{1 + ||x_i||_2 \log n}{ (n/\log n)} &&\text{(By Assumption ~\ref{asm:boundedQuerySpace})} \\
&= \gamma(p_i)
\end{align}
\end{proof}

From Lemma~\ref{lem:sensitivityBounds}, we can now derive an upper bound on the total sensitivity $S(\PP)$: 

\begin{corollary}
\label{cor:SumSensitivities}
The sum of sensitivities over all points, $S(\PP) = \sum\limits_{i \in [n]} s(p_i)$, is bounded above by
\begin{align}
S(\PP) &\leq \sum_{i \in [n]} \gamma (p_i) \leq t = 1 + \log n + \log^2 n \\  
&\leq  \BigO(\log^2 n).
\end{align}
\end{corollary}
\begin{proof}
Leveraging the inequality established by Lemma~\ref{lem:sensitivityBounds}, we have the following upper bound for the sum of sensitivities, i.e., for $S(\PP)$:
\begin{align*}
S(\PP) &\leq \sum_{i \in [n]} \gamma(p_i) \\
&= \sum_{i \in [n]} \left(\frac{1}{n} + \frac{\log n  + ||x_i||_2 \log^2 n}{n} \right) \\
&\leq 1 + \log n + \log^2 n \\
&= \BigO(\log^2 n)
\end{align*}
\end{proof}

Finally, under the combination of the results above, we present the derivation of the $(\eps, \delta)$-FPRAS.

\begin{theorem}[$(\epsilon,\delta)$-Coreset FPRAS]
\label{thm:epsilonCoreset}
Given any $\epsilon, \delta \in (0,1)$ and a data set $\PP$, Algorithm~\ref{algorithm} generates an $\epsilon$-coreset, $(S, u)$, of size 
$$
|\SS| = \BigO\left(\log n /\epsilon^2 (d\log \log n + \log^2 n \log (1/ \delta)\right) = \BigO_{\epsilon,\delta}\left(\log n (d\log \log n + \log^2 n) \right),
$$
for the L2-norm SVM problem with probability at least 1 - $\delta$. Moreover, our algorithm runs in $\BigO(nd)$ time.
\end{theorem}
\begin{proof}
First, we leverage the seminal result by Vapnik et al. \cite{vapnik1998statistical} stating that the VC dimension of a separating hyperplane with a margin $w$, i.e., $\text{dim}(\mathcal{F})$, is bounded above by
$$
\text{dim}(\mathcal{F}) \leq d + 1 = \BigO(d).
$$
Now, by Theorem~\ref{thm:sampling}, we have that the coreset constructed by our algorithm is an $\epsilon$-coreset for the SVM problem with probability at least $1 - \delta$, and the size of the coreset is established by plugging in the bound for the sum of sensitivities from Corollary~\ref{cor:SumSensitivities} to Equation \eqref{eqn:samplesNeeded}:
$$
|\SS| \ge \frac{ct}{\epsilon^2}\left(d\log(t) + t\log(\frac{1}{\delta})\right).
$$
Moreover, note that the computation time of our algorithm is dominated by computing the upper bounds on the sensitivity, i.e., Line~\ref{E2} of our Algorithm, which is in turn a $\mathcal{O}(d)$ time operation per point, yielding a total running time of $\mathcal{O}(nd)$.
\end{proof}

Thus, our coresets are of size polylogarithmic in the number of points $n$ and polynomial in the dimension of the points $d$. Note that since $d \ll n$ in the applications we are considering, the theorem above proves that the coreset generated by our approximation scheme is capable of generating an $1 \pm \epsilon$ approximation to the SVM problem with probability $1-\delta$, even when the coreset size is significantly smaller than the size of the original input points $n$.

\section{Method}
\label{sec:method}
\begin{algorithm}[!htb]
\label{algorithm}
\caption{$\coreset(\PP,\eps,\delta)$\label{one}}
{\begin{tabbing}
\textbf{Input:} \quad\quad\= A set of training points $\PP \subseteq \REAL^d$ containing $n$ points, \\
\>an error parameter $\eps\in(0,1)$, and failure probability $\delta\in(0,1)$.\\
\textbf{Output:} \>An $\epsilon$-coreset $(\SS,u)$ for the query space $\mathcal{F}$ with probability at least $1-\delta$.
\end{tabbing}}

\label{E1} \For{$i \in [n]$}   
{
$\gamma(p_i) \gets \frac{1}{n} + \frac{\log n  + ||x_i||_2 \log^2 n}{n}$ \label{E2} \\ 
}
$t \gets \sum_{i \in [n]} \gamma(p_i)$ \label{E3} \\
Let
\[
m \gets \Omega \left(\frac{t}{\eps^2}\left(d\log t+\log\left(\frac{1}{\delta}\right)\right) \right),
\]\\ \label{E4} 
$(K_1,\ldots,K_n) \sim \text{Multinomial}\left(m, \pi_i = \gamma(p_i)/t \, \, \, \forall{i \in [n]}\right)$  \label{E5} \\
$\SS \gets \left\{ p_i \in \PP \, : \, K_i > 0 \right\}$ \\ \label{E9}
//\textit{Compute the weights $u:\PP \to \mathbb{R}_{\geq 0}$ for every point $p_i \in \SS$.} \\ \label{E6} 
\For{\label{E7} $i \in [n]$}   
{
	$u(p_i) \gets \frac{t  K_i}{\gamma(p_i) |\SS|}$ \\ \label{E8}
}
\Return $(\SS,u)$  \label{E10}
\end{algorithm}

In this section, we will give an overview of the algorithm to compute the coreset $\SS$, which can then be used to train the SVM classifier. We will highlight the important insights of the method, before explaining each step closely.

Our algorithm is based off the idea that for any given dataset $\PP$, we assign an importance $\gamma(p_i)$ to each data point $p_i$ and then sample from the dataset according to the multinomial distribution emerging from this procedure. The crucial insight to this method is how we assign the importances $\gamma(p_i)$. In particular, we use an overapproximation of the sensitivity $s(p_i)$ of each point, i.e., $\gamma(p_i)$, to assign importances, which are obtained from the analysis from the previous section. Following the sampling of points, we further assign weights $u(p_i)$ to each data points, which are proportional to the number of times the point has been sampled. We then train a SVM classifier on the weighted coreset $(\SS,u)$ using any standard SVM library. The resulting algorithm is an $(\eps,\delta)$-FPRAS for approximating the trained classifier.

The overall method to compute the desired coreset is outlined in Algorithm~\ref{algorithm}. Given a set of input data $\PP$, an error parameter $\eps$, and the desired failure probability $\delta$, the algorithm returns an $\eps$-coreset $(\SS,u)$ from the query space $\mathcal{F}$ with probability at least $1-\delta$. In Line~\ref{E2} we compute the importance of a point, i.e., the upper bound on the sensitivity $s(p_i)$ of a point $p_i$, see Lemma~\ref{lem:sensitivityBounds} for more details. In Line~\ref{E4}, we compute the necessary number of samples to include in $(\SS,u)$, according to Theorem~\ref{thm:sampling}, and we then sample from the resulting multinomial distribution, see Line~\ref{E5}. Note that samples in Line~\ref{E8} are weighted according to Theorem~\ref{thm:sampling}.

\section{Results}
\label{sec:results}

We evaluate the performance of our coreset construction algorithm against two real-world, publicly available data sets \cite{Lichman:2013} and compare its effectiveness to Core Vector Machines (CVMs), another approximation algorithm for SVM \cite{tsang2005core}, and categorical uniform sampling, i.e., we sample uniformly a number of data points from each label $y_i$. In particular, we selected a set of $M = 10$ subsample sizes $S_1,\ldots,S_\text{M} \subset [n]$ for each data set of size $n$ and ran all of the three aforementioned coreset construction algorithms to construct and evaluate the accuracy of subsamples sizes $S_1,\ldots,S_\text{M}$. The results were averaged across $100$ trials for each subsample size. Our experiments were implemented in Python and performed on a 3.2GHz i7-6900K (8 cores total) machine.


\subsection{Credit Card Dataset}
The Credit Card dataset\footnote{\url{https://archive.ics.uci.edu/ml/datasets/default+of+credit+card+clients}} contains $n = 30,000$ entries each consisting of $d = 24$ features of a client, such as education, and age. The goal of the classification task is to predict whether a client will default on a payment. 

Figure~\ref{fig:CreditCard} depicts the accuracy of the subsamples generated by each algorithm and the computation time required to construct the subsamples. Our results show that for a given subsample size, our coreset construction algorithm runs more efficiently and achieves significantly higher approximation accuracy with respect to the SVM objective function in comparison to uniform sampling and CVM. 
Note that the relative error is still relatively high since the dataset is not very large and benefits of coresets become even more visible for significantly larger datasets.

As can be seen in Fig.~\ref{fig:CreditCardSensitivity} our coreset sampling process noticeably differs from uniform sampling and some data points are sampled with much higher probability than others. This is in line with the idea of sampling more important points with higher probability.

\begin{figure*}[!htb]
  \centering
  \begin{minipage}{0.49\textwidth}
  \centering
  \includegraphics[width=1\textwidth]{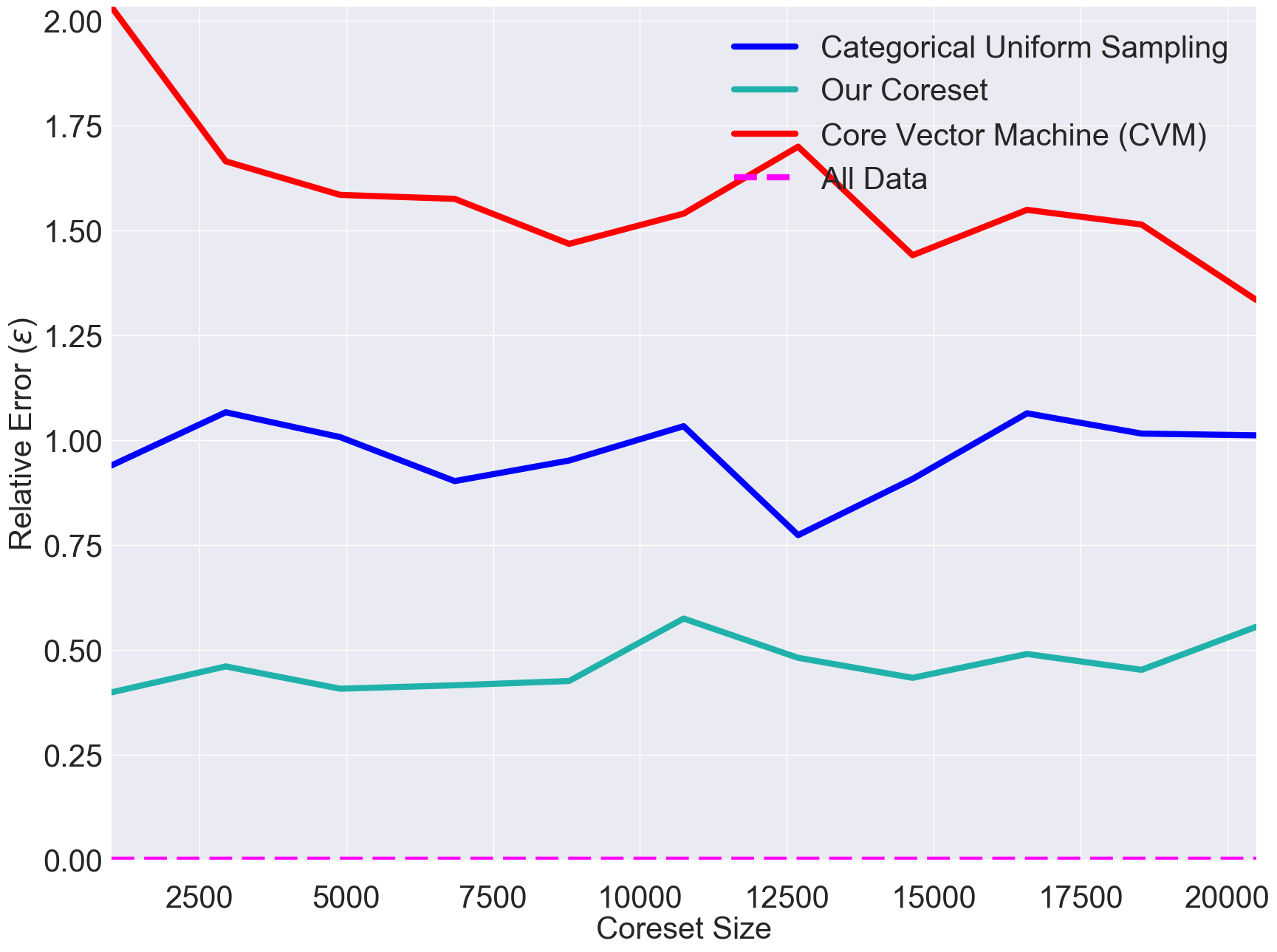} 
  a)
  \end{minipage}
  \begin{minipage}{0.49\textwidth}
  \centering
 \includegraphics[width=1\textwidth]{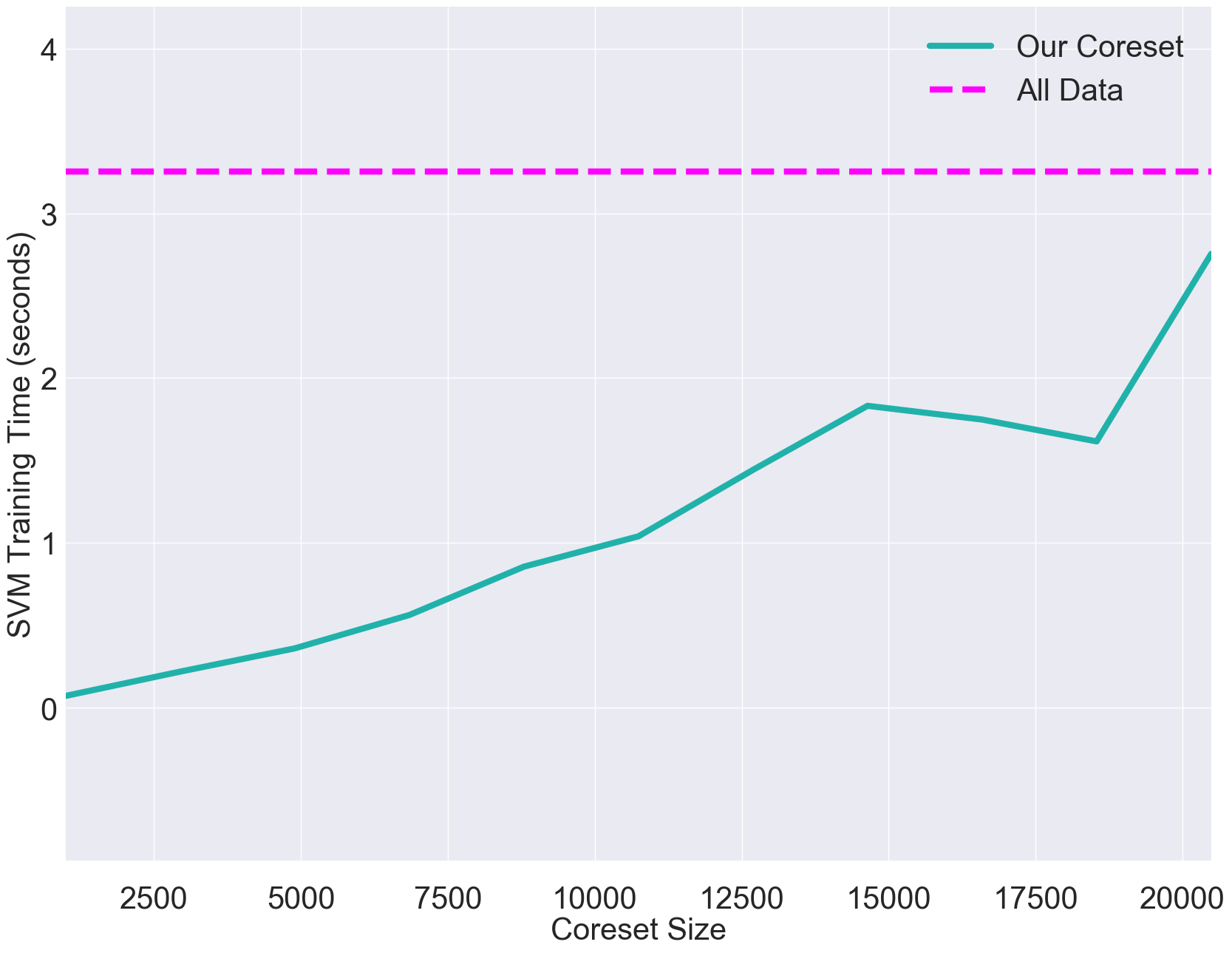}
 b)
  \end{minipage}
  
		\caption{a) Accuracy of each algorithm with respect to the SVM objective function \eqref{eqn:svm}. b) Computation time required to train the SVM using subsamples relative to using all of the input data.}
	\label{fig:CreditCard}
\end{figure*}

\begin{figure*}[!htb]
  \centering
  \begin{minipage}{0.48\textwidth}
  \centering
  \includegraphics[width=1\textwidth]{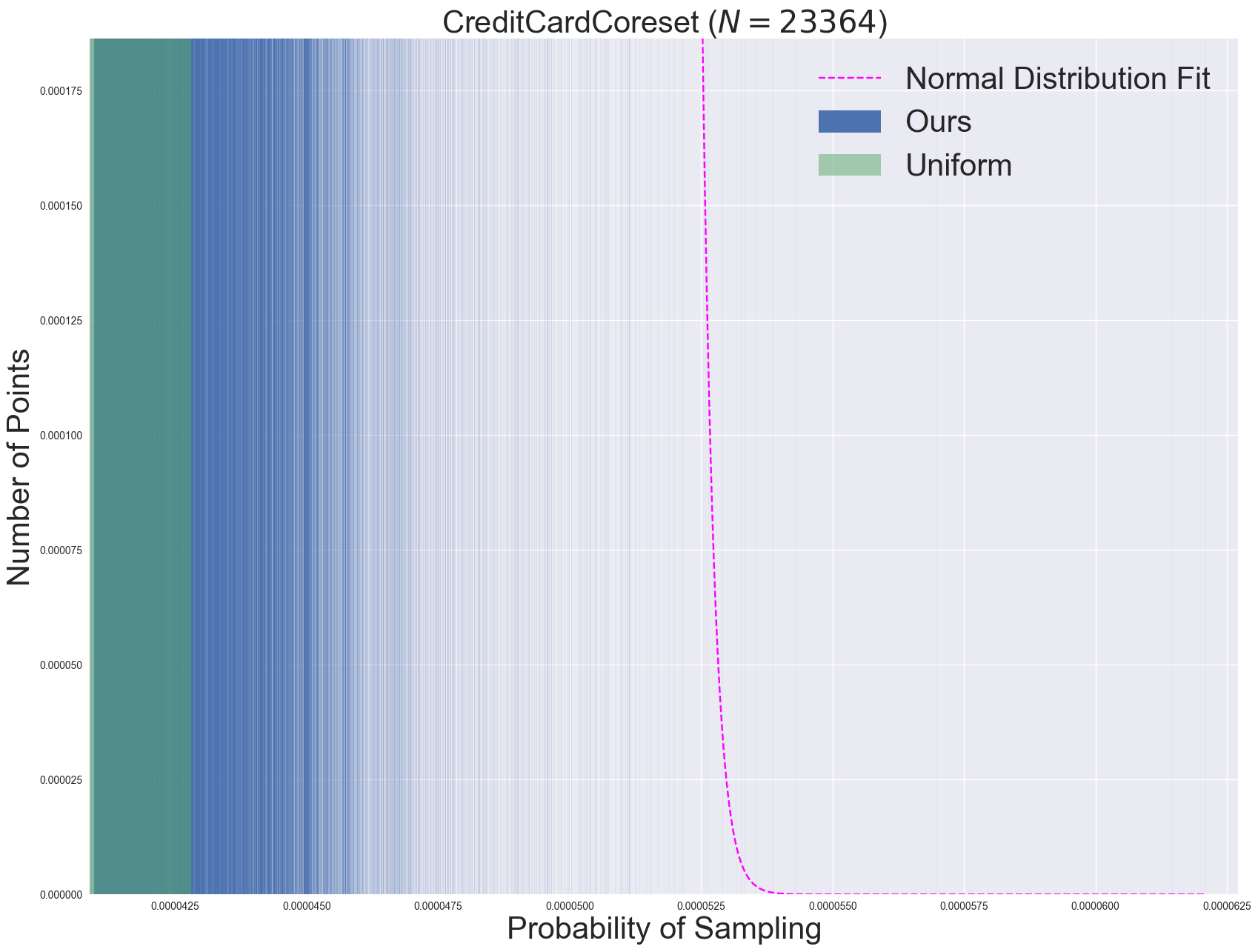}
  a)
  \end{minipage}
  \begin{minipage}{0.48\textwidth}
  \centering
 \includegraphics[width=1\textwidth]{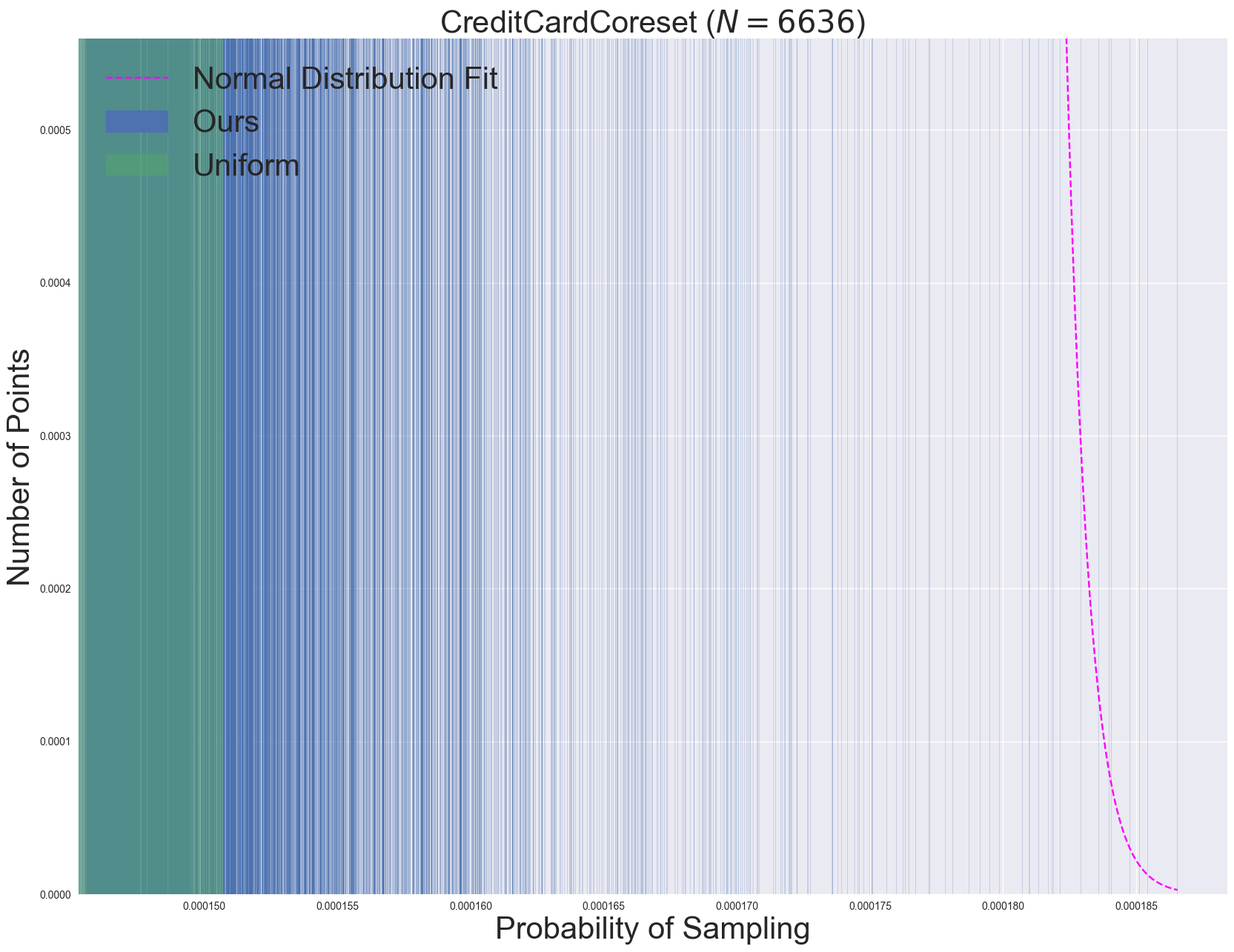}
 b)
  \end{minipage}
  
		\caption{Sampling distribution of points based on the computed upper bounds on sensitivity.}
	\label{fig:CreditCardSensitivity}
\end{figure*}

\subsection{Skin Dataset}
The Skin dataset\footnote{\url{https://archive.ics.uci.edu/ml/datasets/Skin+Segmentation/}} consists of $n = 245,057$ data points with $d = 4$ attributes per point. The attributes include random samples of B,G,R values from face images and the goal of the classification task is to determine whether these samples are skin or non-skin samples. 

Our coreset outperforms uniform sampling for all coreset sizes (cf.~Fig.~\ref{fig:SkinData}), while computation time of coreset generation and SVM training remains a fraction of the original dataset. Due to poor performance CVM is omitted from the showed results. Note that this dataset is significantly larger than the Credit Card dataset and thus the advantages in error and most significantly runtime are more prominent.

\begin{figure*}[!htb]
  \centering
  \begin{minipage}{0.49\textwidth}
  \centering
  \includegraphics[width=1\textwidth]{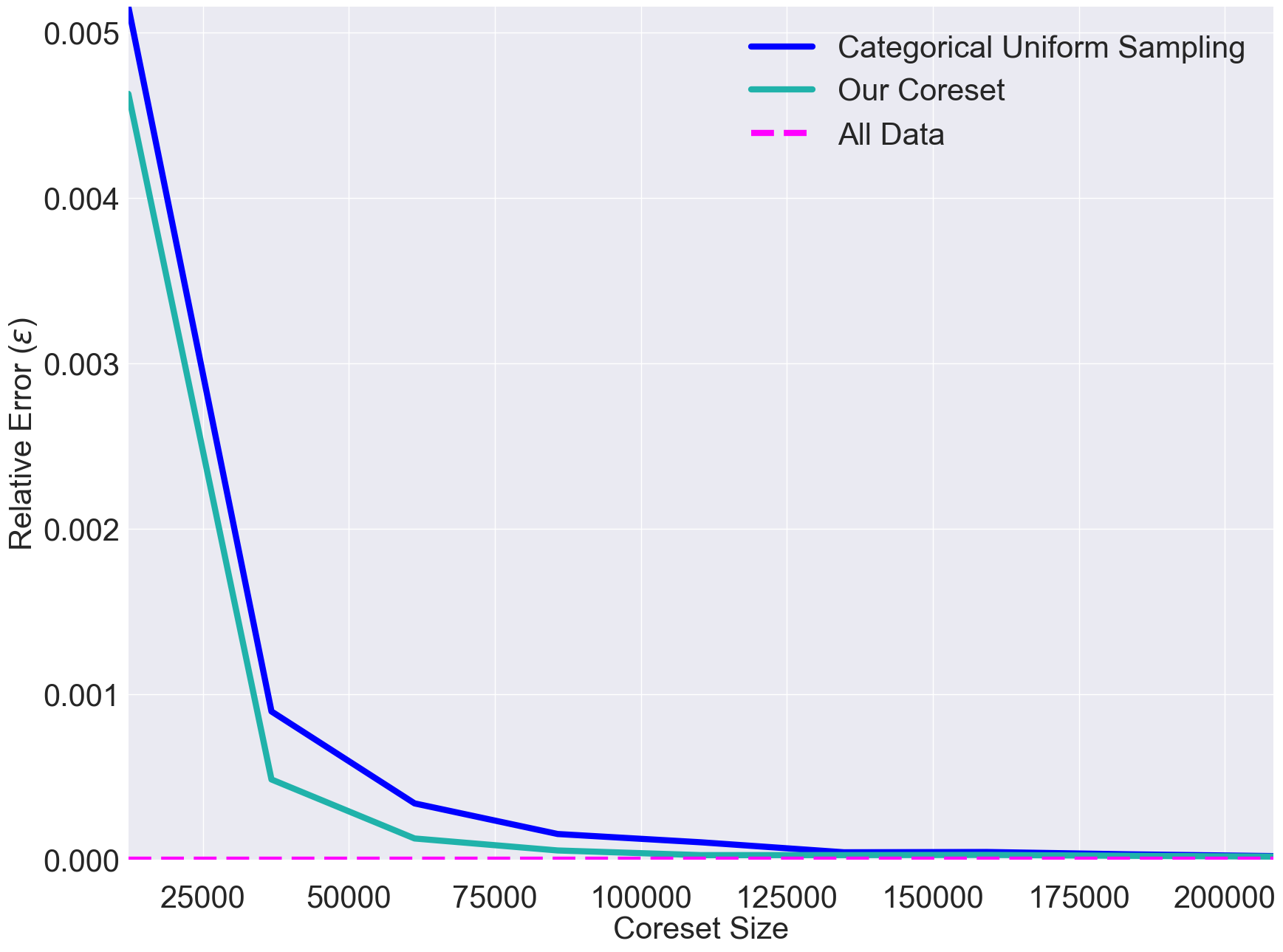} 
  a)
  \end{minipage}
  \begin{minipage}{0.49\textwidth}
  \centering
 \includegraphics[width=1\textwidth]{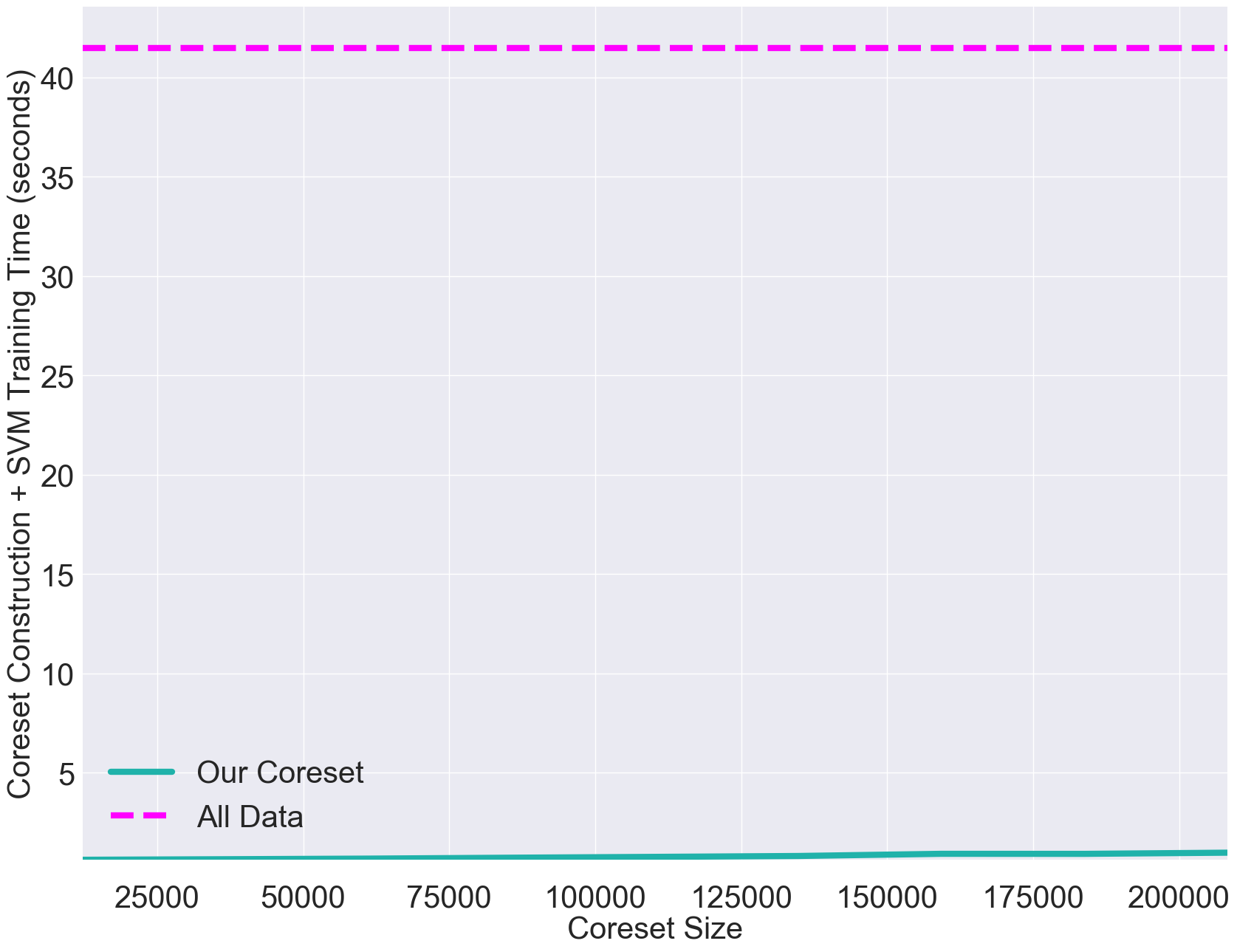}
 b)
  \end{minipage}
  
		\caption{a) Accuracy of each algorithm with respect to the SVM objective function \eqref{eqn:svm}. b) Computation time required to train the SVM using subsamples relative to using all of the input data. CVM is omitted due to poor performance.}
	\label{fig:SkinData}
\end{figure*}

\section{Conclusion}
\label{sec:conclusion}
We presented an efficient coreset algorithm for obtaining substantial speedups in SVM training at the cost of small, provably-bounded approximation error. Our approach relies on the intuitive fact that data is often redundant and that some data points are more important than others. We showed that by obtaining tight bounds on the importance, i.e. sensitivity, of each point, coresets of size polylogarithmic in the number of points and polynomial in the dimension of the points can be efficiently constructed. To the best of our knowledge, this paper presents the first method for constructing coresets of this size that is also applicable to streaming settings, using the merge-and-reduce approach familiar to coresets \cite{braverman2016new}. Our favorable empirical results demonstrate the effectiveness of our algorithm in accelerating the training time of SVMs in real-world data sets. We conjecture that our coreset construction method can be extended to significantly speed up SVM training for nonlinear kernels as well as other popular machine learning algorithms, such as deep learning.


\bibliographystyle{plain}
\bibliography{refs}

\end{document}